\documentclass[times, 10pt]{article} 
\usepackage{latex8}
\usepackage{times}
\usepackage{epsf}
\usepackage{latexsym}
\usepackage{graphicx}
\usepackage{subfigure}
\usepackage{xcolor}

\makeatletter

\newcommand{\comment}[1]{{}}

\newcommand{\NatNumbers}{\mathrm{I\!N}}

\newcommand{\ac}{\mbox{\rm AC$^0$}}

\newcommand{\shac}{\mbox{\rm \#AC$^0$}}

\newcommand{\tc}{\mbox{\rm TC$^0$}}

\newcommand{\acc}{\mbox{ACC$^0$}}

\newcommand{\nc}{\mbox{\rm NC$^1$}}

\newcommand{\halmos}{\rule{1ex}{1.4ex}}
\newenvironment{proof}{\noindent {\bf Proof}:}{\hspace*{\fill}$\halmos$\medskip}

\newtheorem{theorem}{Theorem}

\newtheorem{lemma}[theorem]{Lemma}

\newtheorem{proposition}[theorem]{Proposition}
\newtheorem{fact}[theorem]{Fact}

\newtheorem{definition}{Definition}

\usepackage{latex8}
\makeatother
\begin{document}

\title{Topology inside \nc}

\author{Eric Allender\thanks{Supported in part by NSF grants
CCF-0832787 and CCF-1064785.}\\
Rutgers University\\
New Brunswick, NJ 08855, USA\\
allender@cs.rutgers.edu
\and Samir Datta\thanks{This research was done while this author was
a postdoctoral associate at WINLAB, Rutgers University.}\\
Chennai Mathematical Institute\\ 
Chennai, TN 603 103, India\\
sdatta@cmi.ac.in
\and
Arsenii Karnaukhov\\
Independent Researcher\\
senixka@gmail.com
\and
Grisha Pochuev\\
Independent Researcher\\
n{\_}854@mail.ru
\and
Sambuddha Roy
\thanks{Work done while this author was a graduate student at Rutgers University, NJ, USA}
\\
Qualtrics\\
Seattle, WA  98101, USA\\
shombuddho@gmail.com
\and
Alexander Shekhovstov\\
Columbia University\\
New York, NY  10027, USA\\
alex.v.shekhovtsov@gmail.com
}

\maketitle

\begin{abstract}
We show that $\acc$ is precisely what can be computed
with constant-width circuits of polynomial size
and polylogarithmic genus.  This extends a characterization
given by Hansen \cite{Hansen2006}, showing that 
planar constant-width circuits
also characterize $\acc$.  Thus polylogarithmic genus
provides no additional computational power in this model.
We consider other generalizations of planarity, including
crossing number and thickness.  We show that constant-width
circuits of polynomial size and thickness
two already suffice to capture all of $\nc$.
\end{abstract}

\Section{Introduction}

Circuit complexity provides a vocabulary for classifying and understanding
the complexity of Boolean functions.  The complexity class $\nc$ plays
a central role in the study of circuit complexity; $\nc$ can be defined either
as the class of Boolean functions that can be represented by Boolean formulae
of size polynomial in $n$ (where $n$ is the number of input variables), or as
the class of Boolean functions that can be computed by circuits of depth 
$O(\log n)$, consisting of gates of fan-in $O(1)$.  Our focus in this paper
is on $\nc$ and its subclasses.

The complexity class $\acc$ is one of the most important
subclasses of $\nc$.  Barrington's characterization of
$\nc$ in terms of constant-width branching programs \cite{BarringtonNC1.89}
highlighted the importance of algebraic considerations in
studying small circuit complexity classes, and initiated
a productive line of research reinforcing the connections between
circuit complexity and formal language theory \cite{Barr.Therien.88, bcst, mpt}.
In this framework, computation over {\em non-solvable} monoids
gives complete problems for $\nc$, while computation over
{\em solvable} monoids yields problems in $\acc$.

The class $\acc$ also attracts attention, because it lies
at the frontier of current lower bound techniques.  $\acc$
is the union of the classes $\ac[m]$ of problems computed by
constant-depth polynomial-size circuits of AND, OR, and MOD$_m$
gates.  If $m$ is prime, then $\ac[m]$ is known to be a proper
subclass of $\acc$ \cite{Smolensky.87, Razborov.87}, but for $m$
composite, it remains unknown if EXP is contained in
(nonuniform\footnote{Briefly, a circuit family $\{C_n : n \in \NatNumbers\}$ is
said to be {\em uniform} if there is an efficient algorithm to construct
$C_n$.  Our focus in this paper is on nonuniform circuit families.  For
more background, consult the text by Vollmer \cite{vollmer}.})
$\ac[m]$.  Indeed, it was viewed as a major advance when it was
shown that NEXP is not contained 
in (nonuniform) $\acc$ \cite{williams}; subsequently, Murray and Williams
improved this, to show that nondeterministic quasipolynomial time is not
in $\acc$ \cite{murray.williams}.  Other lower bounds for $\acc$ circuits building on
\cite{williams,murray.williams} can be found in
\cite{chen3,chen1,chen2,papakonstantinou,wang}.

In 2006, Hansen \cite{Hansen2006} 
gave a very surprising characterization of
$\acc$, in terms of constant-width circuits.  Barrington's theorem
\cite{BarringtonNC1.89}
yields as a corollary a characterization of
$\nc$ as precisely the problems solvable by constant-width circuits
of polynomial size.  If NOT gates are allowed, then these circuits can
be made to be planar, but if NOT gates are allowed only at the
leaves (i.e., at the inputs), then Hansen is able to build on
earlier work \cite{Hansen-circuits.02}
to show that $\acc$ is 
precisely the class of languages accepted by polynomial-size
constant-width {\em planar} circuits.  This is a beautiful and
unexpected characterization, making no blatant
reference to counting mod $m$ or to the algebraic considerations 
that have been central to all previous work on $\acc$.  Subsequently, it
was shown that even the more restrictive model of planar nondeterministic
branching programs suffices to recognize all languages in $\acc$ 
\cite{hansen.koucky}.

$\tc$ is an important complexity class lying between $\acc$ and $\nc$.
Since the papers of Hansen and Barrington combine to give characterizations
of $\acc$ and $\nc$ in terms of constant-width circuits, it was natural
to pick up the question of whether $\tc$ also corresponds to a class
of constant-width circuits.  Since we wanted to characterize a complexity
class that is intermediate between $\acc$ and $\nc$, we focused on
graphs that are ``intermediate'' in some sense between planar and
unrestricted.  In this paper, we consider some of the 
most important graph-theoretic notions that generalize the concept
of planarity:
\begin{itemize}
\item Crossing Number
\item Genus
\item Thickness
\end{itemize}
It will suffice for the reader to have an informal grasp of these
notions; we provide a few additional definitions later on where
they are needed.  The {\em crossing number} of a graph is the least
number of ``edge intersections'' required in any embedding of the
graph in the plane.  The {\em genus} of a graph is the least
number of ``handles'' (or ``doughnut halves'') that need to be
attached to the plane in order to provide a surface on which the
graph can be embedded with no edge crossings.  The {\em thickness}
of a graph is the smallest number of blocks needed in a partition
of the edge set, so that the vertices can be embedded in a plane so
that none of the edges in any block of the partition cross each other;
intuitively this is the number of transparencies that would be necessary in
order to represent
the graph, where each transparency is planar.
(For more complete definitions, please consult a graph theory text, such 
as \cite{Gibbons.85}.)
Planar graphs have crossing number 0, genus 0, and thickness 1.
For any graph $G$, 
thickness$(G) - 1 \leq$ genus$(G) \leq$ crossing.number$(G)$. 

Our main theorem is that constant-width polynomial size
circuits of polylogarithmic genus compute exactly the problems
in $\acc$.  As a corollary, the same is true for circuits
with polylogarithmic crossing number.  In contrast, constant-width
circuits of thickness two already suffice to compute all problems
in $\nc$.  We can view this as a positive result, because it yields
additional information about $\acc$ and $\nc$.  However, it is also
in some sense a negative result, in that it removes the most prominent
candidates for a possible characterization of $\tc$ in terms of
constant-width circuits.  We leave the task of finding such a
characterization as our main open problem.

If $C$ is a circuit whose graph has genus $k$, we will say that $C$ has
genus $k$.

A preliminary version of this work appeared as \cite{ccc,eccc}.  However, the
proof of the main theorem presented in the preliminary version is
incorrect, as discussed in \cite{parting.shots}, where the task of
finding a correct proof was listed as an open question.  The fourth author, and (independently) the third and sixth authors of the current article found a quite simple proof, with some assistance from ChatGPT. (See, e.g., \cite{pochuev}.)  In modified form, that is the proof that is presented in Section~\ref{mainsec}.

\Section{Definitions and Preliminaries}\label{defsec}

We first define a layered digraph : 
\begin{definition}
We call a digraph {\em layered} if there is a partition of the 
vertex set into sets $V_0, V_1, \cdots , V_l$ (and we call them
{\em layers} or {\em levels}) where every (directed) edge in the graph is from 
some layer $V_i$ to $V_{i+1}$. 
\end{definition}

\begin{definition}
The {\em width} of a layered digraph with layers $V_0, \ldots, V_r$
is max$\{|V_i| : 0 \leq i \leq r\}$.
\end{definition}

A circuit of width $w$ is a layered digraph of width $w$ where each vertex
is labeled either as an AND gate, an OR gate, an input
variable $x_i$, or a negated input variable $\neg x_i$.  It
is important to note that inputs can appear on any level,
and inputs can appear more than once.

Hansen's characterization of $\acc$ in terms of planar constant-width circuits
is the starting point for our investigation:

\begin{theorem}\label{hansenthm}\cite{Hansen2006}
  $\acc$ is equal to the set of languages accepted by constant-width planar circuits of polynomial size.
\end{theorem}

\begin{definition}
  If $C$ is a layered circuit with layers $V_0, \ldots V_r$, then $C[a, b]$ denotes the subcircuit of $C$ spanned by layers $V_a \ldots, V_b.$
  \end{definition}

A circuit is planar if it can be embedded in the plane with no
two edges crossing.  More generally, a circuit has genus $\leq k$
if it can embedded on a surface of genus $k$ with no edges crossing.
The reader need have no detailed understanding of the topological
notion of genus; an informal grasp of the topic is sufficient.
Informally, a graph has genus $k$ if it can be embedded with no
edge crossings on a plane with $k$ ``handles'', where a ``handle''
is a bent cylinder that is attached to the plane at each end.

\begin{definition}
  The genus of $G$, denoted $\gamma(G)$ is the least number $k$
  such that $G$ can be embedded in a surface of genus $k$.  Thus a
  planar graph has genus 0.
  \end{definition}

The preliminary version of this work \cite{ccc,eccc} contained a
detailed discussion of certain properties that an
embedding of a circuit onto a surface of genus $k$ could be
assumed to have, summarized as \cite[Theorem 1]{ccc}.  Although
it is possible that these properties could be useful for
future work, they are not needed for the simplified proof of
our main theorem, and thus this material has not been included
in this revision.

Instead the main fact that we need about genus is summarized in the
following theorem from \cite[Corollary 2]{BattleHararyKodamaYoungs1962}:
\begin{theorem}\label{thm:battle}
Suppose the graph $G$ consists of connected components $G_1, \ldots, G_k$. Then, 
$$
\gamma(G) = \gamma(G_1) + \ldots + \gamma(G_k).
$$
\end{theorem}

\Section{Small Genus Characterizes $\acc$}\label{mainsec}

\begin{theorem}\label{mainthm}
Let $A$ be a language.  $A$ is in $\acc$ if and only if
$A$ is accepted by a family of constant-width circuits
of polynomial size and polylogarithmic genus.
\end{theorem}
\begin{proof}
One direction follows immediately from Hansen's characterization
\cite{Hansen2006} where the genus is even required to be zero.

For the other direction, we first require the following technical
lemma:

\begin{lemma}
\label{lem:select}
Let $C$ be a layered circuit of genus $g$, with layers $V_0, \ldots, V_r$.
Then, there is a number $q \leq 3g + 10$ and indices $0=t_1 < t_2 < \ldots < t_q = r$, such that for any $i=1\ldots q-1$,
    \begin{itemize}
        \item either $C[t_i, t_{i+1}]$ is planar, or
        \item $t_{i+1} = t_i + 1$.
    \end{itemize}

\end{lemma}
\begin{proof}
    We construct the sequence $t_i$ greedily. For $i=2, 3, \ldots$ define $t_i > t_{i-1}$ to be the largest number such that $C[t_{i-1}, t_i]$ is planar. If even $C[t_{i-1}, t_{i-1}+1]$ is not a planar graph, then set $t_i =  t_{i-1}+1$. We proceed until step $i = q$ when we reach $t_q = r$.

    We claim that $q \leq 3g+10$.  To see this assume, $q > 3g + 10$, and we will derive a contradiction. By definition, $C[t_i, t_{i+1} + 1]$ is not planar for any $i < q - 1$.
    For $3i + 1 < q$, the graphs $C[t_{3i}, t_{3i + 1} + 1]$ are disjoint and are not planar. Therefore, by Theorem~\ref{thm:battle} we have genus of $C$, $\gamma(C) \geq \sum_i \gamma(C[t_{3i}, t_{3i + 1} + 1]) > g$ (since removing edges does not increase genus, and since each $C[t_{3i}, t_{3i + 1} + 1]$ has genus at least $1$).  This is the desired contradiction.
\end{proof}

We now continue with the proof of the main theorem.  Let $\{C_n: n \in \NatNumbers\}$ be a family of circuits of polynomial size with width $w$ and genus $g=O(\log^c n)$.  Consider a
given circuit $C_n$.
Let $V_0, \ldots, V_r$ be the layers of $C_n$. Let $0=t_1 < t_2 < \ldots < t_q = r$ be given by Lemma~\ref{lem:select}. For any $i < q$, $C[t_i, t_{i+1}]$ can be computed by $\acc$ circuits.
This is because either $C[t_i, t_{i+1}]$ is planar (in which case this claim follows from Theorem~\ref{hansenthm})
or $t_{i+1} = t_i + 1$ holds and $C[t_i, t_{i+1}]$ can even be computed by $\ac$.  For every $i$ such that $C[t_i, t_{i+1}]$ is planar, we use Hansen's theorem to convert it into a constant-depth AND, OR, $\textsc{MOD}_m$ circuit $C_i'$ (not necessarily planar).
It is important to note that the modulus $m$ depends only on the width of the circuits in the circuit family, and thus the same modulus $m$ can be used for all of the subcircuits of $C_n$, for every $n$.  This is implicit in \cite{Hansen2006}, but it is appropriate to say a bit more about this point here.  In \cite{Hansen2006}, Hansen shows that simulating planar circuits of width $w$ can be reduced to simulating {\em cylindrical} circuits of width no more than $w$, and then appealing to a theorem of \cite{Hansen-circuits.02}, in which it is shown that cylindrical circuits of width $w'$ compute only functions in $\acc$ (where the proof of this inclusion uses modular gates of modulus $m$, where the modulus depends only on the width $w'$).  Since there are $O(1)$ different choices of $w' \leq w$, this possibly gives rise to $O(1)$ different moduli $m_1,\ldots m_c$, but (as observed in \cite{Smolensky.87}), computation mod $m_i$ can be simulated with a MOD$_m$ gate, for $m=\prod_{i=1}^c m_i$.  

Thus there is a function $f : \{0, 1\}^n \times \{0, 1\}^w \times [q - 1] \to \{0, 1\}^w$ computable in $\acc$, 
where $f$ takes as input a triple $(x,v,i)$ and outputs a string $z$
such that $v$ and $z$ are bit strings of length $w$, having the
property that if the $w$ gates at the start of segment $i$ have the
values given by the vector $v$ and the string $x$ is used to provide
values to the input gates appearing in segment $i$, then the gates at the
output level of segment $i$ will take on the values given by the vector $z$.  This is because, using
$\ac$ circuitry, we can use the string $i$ to select the $\acc$ circuit for segment $i$,
and then use $x$ and $v$ to compute the desired output $z$.

Thus, given $x$, in $\acc$ we can build a layered
graph with width $2^w=O(1)$ and polylogarithmically many
levels, such that there is an edge from node $v$ in level $j$ to node $z$ in 
level $j+1$
if and only if $f(x,v,j) = z$. More precisely, when we say that we can ``build'' such a graph, we mean that there is a set of predetermined vertices and for each pair of vertices $(u, v)$ we compute the value of a Boolean variable $b_{(v, u)}$ that represents existence of the edge $u \to v$.  Furthermore, input $x$ is accepted if and only if there is a path in this graph from the string $v$ in the first level encoding the initial gate values in the circuit, to a string $z$ in the final level where the output gate has value 1. Finding paths in such graphs can be done
in $\ac$.  


Let us describe how to check if there exists a path between two vertices in such a graph $G$ using $\ac$ circuitry. If two vertices $v$ and $u$ are $l$ layers apart, then there are at most $2^{wl}$ possible paths between them $v=v_1, \ldots, v_l=u$. We can build a disjunctive normal form formula $\lor_{v = v_1, \ldots, v_l=u} (b_{(v_1, v_2)} \wedge b_{(v_2, v_3)} \wedge \ldots \wedge b_{(v_{l-1}, v_l)})$ that checks if there exists a path between $v$ and $u$. Set $l = \Theta(\log n)$, then this formula has depth $2$ and polynomial size. Denote by $V_0, V_1, \ldots, V_q$ the layers of $G$, and wlog assume that $q$ is divisible by $l$. For every $i$ and $u \in V_{il}, v \in V_{(i+1)l}$ we compute if there is a path between them. Consider a modified graph $G'$ obtained from $G$ by leaving only layers $V_0, V_l, V_{2l}, \ldots, V_{q}$ and connecting two vertices in adjacent layers if and only if there is a path between them in $G$. Note that the number of layers in $G'$ is $q/l + 1$. We have reduced our problem to checking if there is a path between two vertices in $G'$. let us recursively check the existence of this path. Since the initial number of layers is $O(\log^c n)$ and $l = \Theta(\log n)$, the recursion depth is $O(c)$. On each level of recursion the number of pairs of vertices is polynomial and the DNF for computing existence of a path between is polynomial as well. Therefore, the final circuit has polynomial size and constant depth.

\end{proof}

\Section{A new characterization of \nc}\label{ncosec}

Genus is just one of several possible generalizations of
planarity.  In this section we consider {\em thickness}, and
we show that all problems in $\nc$ can be solved by constant-width
polynomial-size circuits of thickness two.  We actually prove
a stronger result showing that a very limited type of circuit
with thickness two suffices for this task.  First recall Barrington's characterization of $\nc$:
\begin{theorem}[\cite{BarringtonNC1.89}]
\label{thm:bnc1}
    $\nc$ is the class of languages accepted by constant-width polynomial-size circuits with AND, OR, NOT gates.
\end{theorem}


Consider Figure \ref{threepages}, showing three half-planes joined
at a common intersecting line called the {\em spine}.  This is the
type of surface on which we will embed our constant-width circuits,
with the restriction that no wire goes through more than one page and the subgraph on any one half-plane is
{\em upward planar}. (A layered circuit is {\em upward planar} if, when the layers are arranged in columns with layer 0 on the left and the output layer on the right, all edges from any layer $i$ to $i+1$ can be embedded as a straight line segment, with no edge crossings.  The more general planar circuits that characterize $\acc$ allow edges embedded as curves that loop around the initial and final layers.  It is known that constant-width upward planar circuits of polynomial size characterize $\ac$ \cite{Barringtonmonotone.99}.) 
It is easy to see that if only two pages are used, then the entire
graph is upward planar, and by \cite{Barringtonmonotone.99} such circuits
can compute only languages in $\ac$. Let us show that any graph that can be embedded on three pages in this way has thickness two.

 \begin{lemma}
    Any graph embedded in three pages, such that every edge is within only one page, has thickness two.
\end{lemma}\begin{proof}
Consider a graph embedded on three pages. Let us rotate the third page around the spine until it coincides with the first page. By doing small perturbations we can assume that no two different vertices in the third page and first page coincide. Now, let us color the edges on the first and second page in red, while the edges on the third page in black. By definition, the red edges do not cross and the black edges do not cross. Therefore, the thickness of the graph is at most $2$.
\end{proof}

\begin{figure} 
\centerline{\includegraphics[width=7cm]{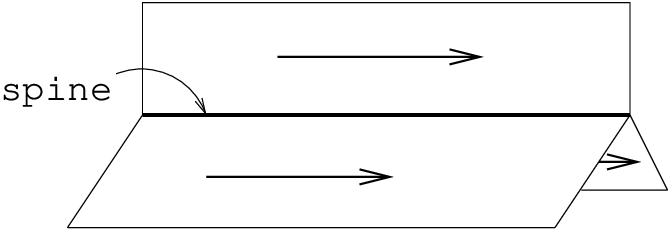}}
\caption{A circuit on three pages}
\label{threepages} 
\end{figure}

Now the question arises as to what 
happens when we allow more than two pages in our
circuit.  Define $k${\em Pages} to be the class of languages 
captured by constant-width polynomial-size circuits on $k$ pages. Recall that we allow a gate to be an OR, an AND, or an input literal, and there can be multiple occurrences of the same input literal. Theorem~\ref{thm:bnc1} implies
$k${\em Pages} $\subseteq \nc$ for every $k$.   

We show that three pages suffice:

\begin{theorem}\label{ncpages}
$\nc = 3${\it Pages}.
\end{theorem} 

In order to present our characterization of $\nc$ in terms of
constant-width circuits on three pages, it is useful to
define a simple nonuniform model of computation:
 
\begin{definition}
Define {\em stacks}$(a,b,c)$ to be the class of languages
accepted by machines with three pushdown stores (with heights bounded
by $a$, $b$, and $c$, respectively) and a computation register.
Only binary values can be stored on the stacks.
The program for the machine consists of a sequence of instructions
(one instruction for each time step), from the following:
\begin{enumerate}
\item push the register value into a stack;
\item pop the topmost entry from a stack into the register; 
\item copy the topmost entry from a stack into the register (without removing it from the stack);
\item discard the topmost entry of a stack; 
\item compute the $\vee$ or $\wedge$ of the topmost entries of two 
stacks and store it in the register; or
\item compute the $\vee$ or $\wedge$ of the topmost entry of a stack
with an input literal and store it in the register.
\end{enumerate}
Notice that the machine is oblivious in that the stack(s) and literal 
corresponding to an instruction are independent of the input.
The output of the machine is the final value that is stored in the register,
and we restrict the running time to be polynomial in the length of
the input.
\end{definition}

See Figure~\ref{canonical} which illustrates stacks(3,3,2).  Our main
simulation using this model (in Lemma~\ref{ncstacks}) involves
permuting five values stored in the stacks as displayed in 
Figure~\ref{canonical}.  Manipulating these values will be accomplished
using stack heights (3,3,2).

\begin{figure} 
\centerline{\includegraphics[width=4cm, angle=270]{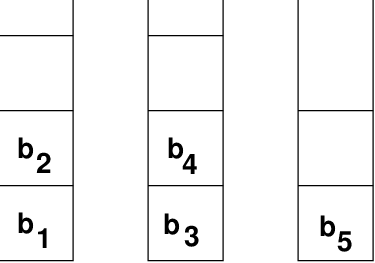}}
\caption{Canonical placement of bits in a stacks(3,3,2) computation}
\label{canonical} 
\end{figure}

Let us show that oblivious computation with 3 stacks can be simulated by 
computation with 3 pages. The height of the stacks corresponds to the width 
of the pages with the spine serving as the register.  The stack is placed on the page so that its top element is closer to the spine than its bottom element. Notice that popping a bit 
corresponds to copying it toward the 
spine while pushing is the reverse (note that this can be accomplished by using single-variable AND gate). An operation such as 
the $\vee$ of the topmost
bits of two stacks can be simulated by bringing the corresponding bits toward 
the spine where the $\vee$ operation is performed.  After an operation is performed we copy the contents of the stacks one layer down and proceed with the next operations.  See the illustration Figure~\ref{fig:stacks3pages} for details.

\begin{figure}[htbp]
\centerline{\includegraphics[width=9cm]{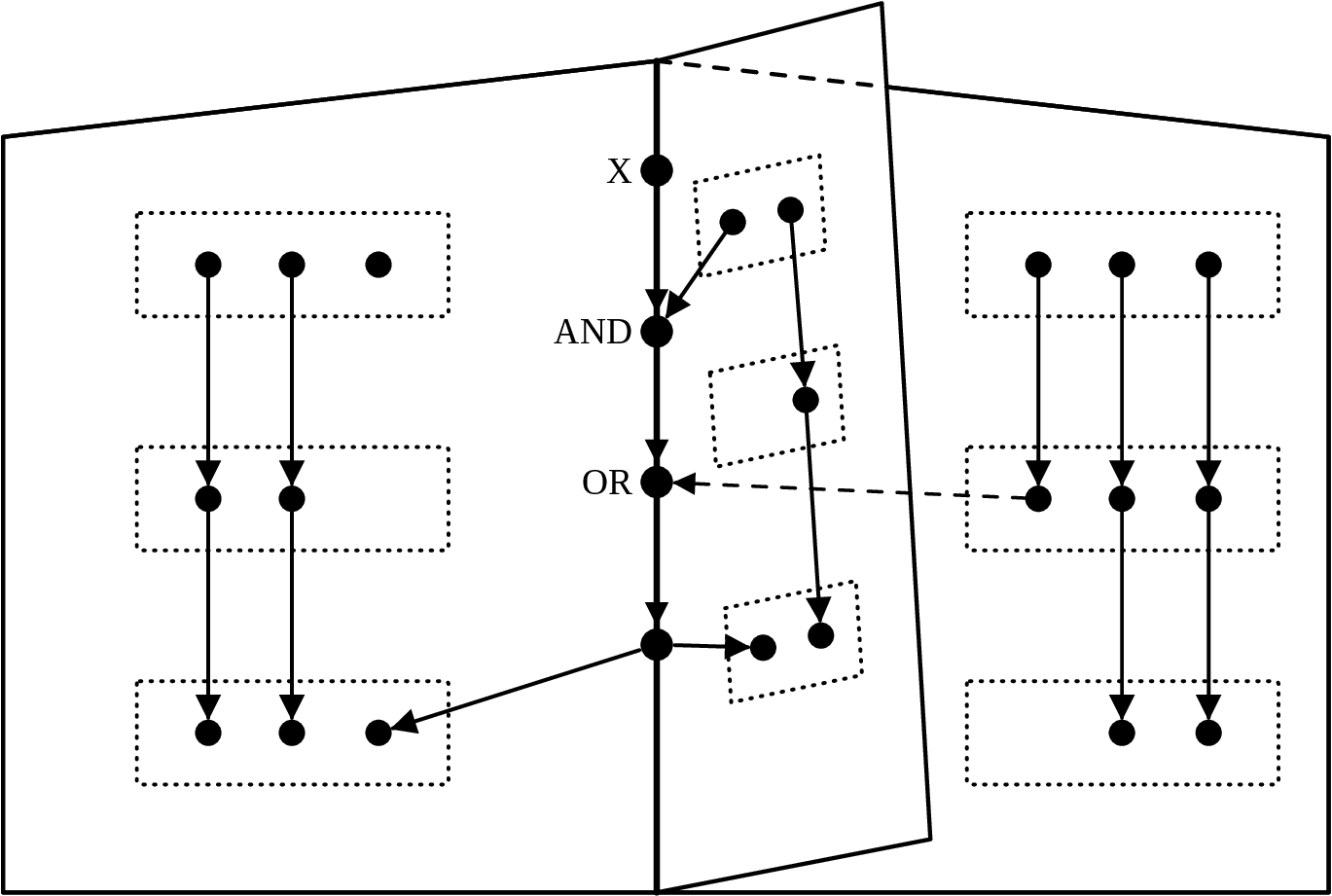}}
  \caption{An illustration of how  stacks$(O(1),O(1),O(1))$ are realized on $3$ pages.}
\label{fig:stacks3pages}
\end{figure}

Therefore we have the 
following proposition:
\begin{proposition}\label{stackspages}
{\em stacks}$(O(1),O(1),O(1)) \subseteq 3${\it Pages}.
\end{proposition}

We have already observed that one inclusion is trivial:
\begin{fact}\label{pagesnc}
$3${\em Pages} $\subseteq \nc$.
\end{fact}

Thus the following lemma will show that $\nc$, 3{\it Pages} and
stacks(3,3,2) all coincide.

\begin{figure} 	
	\begin{center}
			\subfigure{\includegraphics[width=9cm,angle=270]{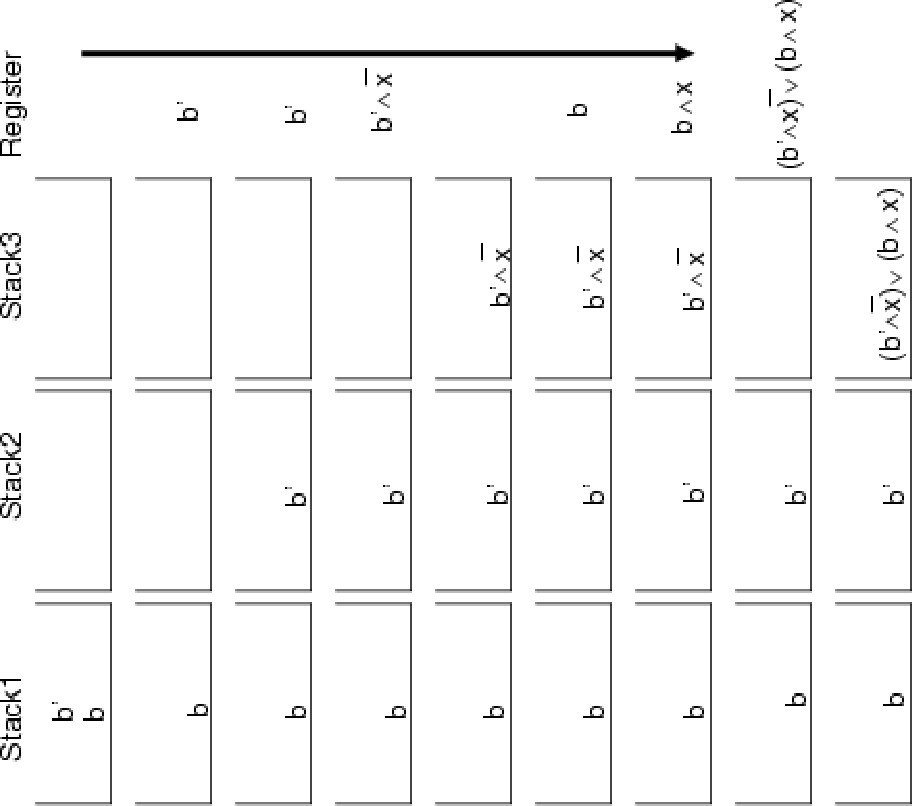}}
			\subfigure{\includegraphics[width=8.7cm,angle=270]{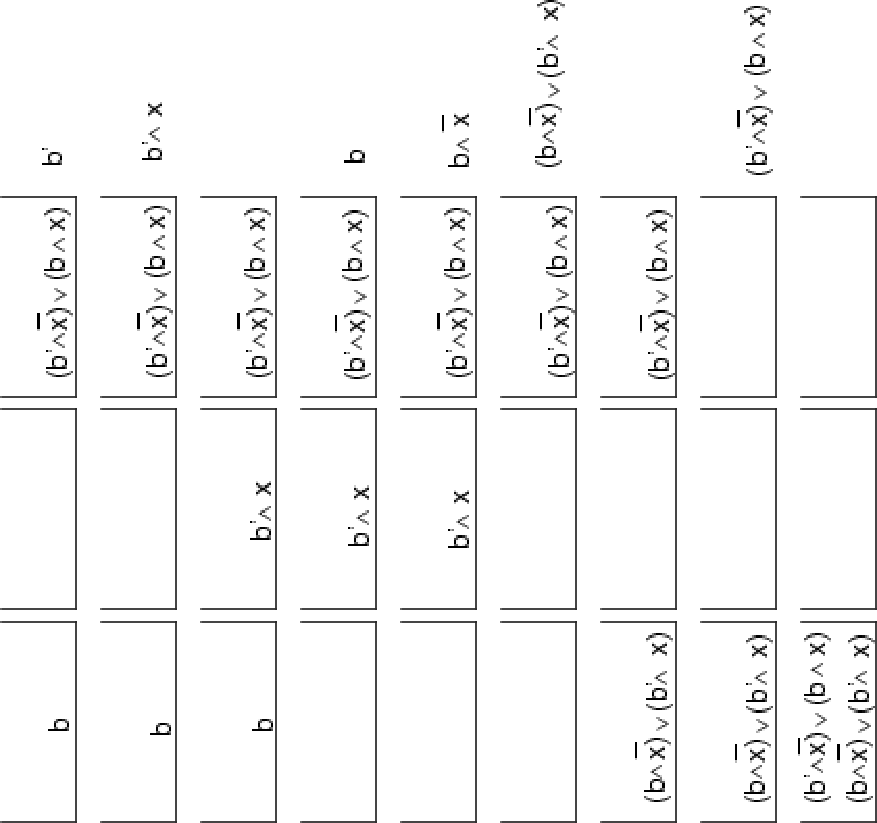}}
	\caption{\label{figstacks}A transposition (if input bit $x$ = 1) and identity 
transformation (if $x$ = 0)}
	\end{center}
\end{figure}

\begin{lemma}\label{ncstacks}
$\nc \subseteq$ {\em stacks}$(3,3,2)$.
\end{lemma}
 
\begin{proof}
Consider Barrington's proof that languages in \nc\ can be computed by 
permutation branching programs over $S_5$ \cite{BarringtonNC1.89}. 
The state of the branching program at some stage of the computation
can be represented as a sequence of five bits ($b_1, b_2, b_3, b_4, b_5$), 
where $b_j = 1$ if the execution of the branching program has led from
the start state to state $j$; since the branching program in Barrington's
simulation is deterministic, exactly one of the bits $b_j$ will be set to 1
at any stage.  We represent the initial state by (1,0,0,0,0).
A permutation branching program consists of a sequence of instructions
$(x_i,\theta_0,\theta_1)$, with the interpretation that if $x_i$ is equal
to 0, then the representation of the state of the program after the
next step is obtained by permuting the five bits ($b_1, b_2, b_3, b_4, b_5$)
according to permutation $\theta_0$, whereas permutation $\theta_1$ is used
if $x_i = 1$.  Without loss of generality, we may assume that one of 
$\{\theta_0,\theta_1\}$ is the identity permutation $\iota$ (because 
instruction $(x_i,\theta_0,\theta_1)$ can be simulated by
$(x_i,\theta_0,\iota)(x_i,\iota,\theta_1)$.  We may further assume that each
$\theta_b$ is a transposition, since each permutation can be represented
as a product of transpositions.  To determine if an input word is accepted,
it suffices to check if $b_1 = 1$ after the last instruction is executed.

The idea of the proof is that we put the five bits in the three
stacks using heights
$2,2,1$ respectively in a canonical way, as illustrated in 
Figure~\ref{canonical}.  Each instruction is of the form $(x_i,\tau,\pi)$
where one of $\{\tau,\pi\}$ is a transposition and the other is $\iota$.
We focus attention on two types of instructions $(x_i,\tau,\pi)$:
\begin{description}
\item[Type 1] The transposition swaps the bits that appear on the tops of
two different stacks.
\item[Type 2] The transposition swaps the bits that appear in
one stack.
\end{description}
Note that every instruction can be simulated by a sequence of instructions
of Type 1 and Type 2.  (For instance, in Figure~\ref{canonical}, to swap bits $b_1$ and $b_3$, we first
use Type 2 instructions to invert stacks 1 and 2, then use a Type 1 instruction
on stacks 1 and 2, and then again invert stacks 1 and 2.)

\begin{lemma}\label{stacksaux}
\begin{enumerate}
\item\label{swapadj} A Type 1 instruction on two stacks
can be accomplished by using exactly one extra place in the other stack.
\item\label{swaptop} A Type 2 instruction on one stack can be accomplished
by using exactly one more place in each of the other two stacks.
\end{enumerate} 
\end{lemma}
\begin{proof} Figure~\ref{figstacks} illustrates the proof of 
Lemma~\ref{stacksaux}.\ref{swaptop}.
A Type 2 instruction decides whether to interchange two bits
$b,b'$ or leave them alone, depending on the value of the literal
$x$.  This corresponds to replacing the sequence $(b',b)$ on one
stack, with the sequence
$((\bar{x}\wedge b) \vee (x \wedge b')$, $(\bar{x}\wedge b') \vee (x \wedge b))$.
This is implemented by computing each of the two expressions in
succession and (during the computation of the second) removing the bits $b,b'$.

The proof of Lemma~\ref{stacksaux}.\ref{swapadj} can be accomplished by
starting with the third step in Figure~\ref{figstacks} (and changing the last
step, so that it moves the contents of the register to Stack 2).
\end{proof}

Lemma~\ref{ncstacks} now follows immediately from Lemma \ref{stacksaux}, since
an operation of Type 2 on the first stack requires heights $2,3,2$, and
an operation of Type 2 on the second stack requires heights $3,2,2$,
for a maximum of $3,3,2$.  Any other operation requires less overhead.
\end{proof}

\Section{Discussion}
Our initial goal of finding a characterization of $\tc$ in terms
of constant width circuits remains a challenge for future work.
It is worth mentioning some approaches that seem not to work.  
The results in this paper seem to rule out characterizations in terms
of crossing number, genus, or thickness.
Algebraic approaches to circuit complexity tend to mimic the structure
of regular sets, and it is known that any regular set that is not complete
for $\nc$ lies inside $\acc$ \cite{BarringtonNC1.89}, \cite{Barr.Therien.88};
thus this avenue does not seem promising when searching for a characterization
of $\tc$.
One might attempt to follow the approach of \cite{aad} by considering
arithmetizations of Boolean circuits.  However, a width-two planar branching
program is presented in \cite{aabdl} for which the problem of counting the
number of accepting paths is hard for $\nc$ under $\acc$ reductions; this
rules out many approaches that one might try in searching for a 
characterization of $\tc$.  On the positive side, a characterization of
$\shac$ (and hence of $\tc$) in terms of counting paths in a restricted
class of branching programs is presented in \cite{aabdl}.  This 
characterization has been extended, to give a characterization of
arithmetic $\nc$ in terms of a class of log-width planar branching programs
\cite{mahajan.rao}.

We find it somewhat intriguing that the graph-theoretic notion of
genus is linked (via Theorem \ref{mainthm} and \cite{Barr.Therien.88})
to the algebraic notion of solvability.  It would be interesting to know
if there are deeper reasons for this linkage.

In parting, we should also mention a result of Hansen \cite{Hansen.ccc.08}
where he proves that {\em quasipolynomial} size constant width nondeterministic
branching programs exactly capture {\em quasipolynomial} size $\acc$. Hansen's result immediately 
implies that quasipolynomial size $\acc$ may either be characterized by quasipolynomial
size polylog genus constant width circuits or by polylog genus constant width (nondeterministic)
branching programs.  All of these characterizations hold in the nonuniform
setting.  A number of obstacles would need to be overcome, before a 
corresponding result can be proven for uniform circuit complexity.

\Section{Acknowledgments}
We thank Robin Thomas, Dan Archdeacon, Bojan Mohar, Bill Steiger, 
Meena Mahajan,
Kasturi Varadarajan, and Carsten Thomassen for their generous help
in answering our questions about graph genus.
A preliminary version of this research (with an erroneous proof of
Theorem \ref{mainthm})
was presented at the 
20th Annual IEEE Conference on Computational Complexity (CCC), in 2005.

\bibliographystyle{latex8}
\bibliography{submission}
\end{document}